\documentclass[11pt]{article}
\usepackage[T1]{fontenc}
\usepackage{lmodern}
\usepackage{geometry} 
\usepackage{amsmath,amssymb,amsthm,amsfonts,dsfont,mathtools}
\usepackage[linesnumbered,ruled,vlined]{algorithm2e}
\usepackage{hyperref}
\usepackage[capitalize]{cleveref}
\usepackage{xcolor}
\usepackage{thmtools}
\usepackage{thm-restate}
\usepackage{multirow}
\usepackage{array}
\crefname{theorem}{Theorem}{Theorems}
\Crefname{lemma}{Lemma}{Lemmas}
\Crefname{figure}{Figure}{Figures}

\newtheorem{theorem}{Theorem}
\newtheorem{lemma}{Lemma}

\newtheorem{claim}{Claim}
\newtheorem{definition}{Definition}
\newtheorem{corollary}{Corollary}

\newcommand{\ceil}[1]{\left\lceil #1 \right\rceil}

\newcommand{\strongcd}{\mathsf{Strong}\text{-}\mathsf{CD}}
\newcommand{\sendercd}{\mathsf{Sender}\text{-}\mathsf{CD}}
\newcommand{\receivercd}{\mathsf{Receiver}\text{-}\mathsf{CD}}
\newcommand{\nocd}{\mathsf{No}\text{-}\mathsf{CD}}

\newcommand{\idle}{\mathsf{idle}}
\newcommand{\listen}{\mathsf{listen}}
\newcommand{\transmit}{\mathsf{transmit}}

\newcommand{\silence}{\mathsf{silence}}
\newcommand{\noise}{\mathsf{collision}}

\newcommand{\ID}{\mathsf{ID}}
\newcommand{\rank}{\mathsf{rank}}

\newcommand{\poly}{\operatorname{poly}}

\newcommand{\LOCAL}{\mathsf{LOCAL}}

\newcommand{\BC}{T_\mathsf{BC}}

\title{Near-Optimal Time-Energy Trade-Offs \\ for Deterministic Leader Election}

\author{
Yi-Jun Chang\thanks{Supported by  Dr.~Max R\"{o}ssler, by the Walter Haefner Foundation, and by the ETH Z\"{u}rich Foundation.}\\
\small ETH Z\"{u}rich \\
\and
Ran Duan\thanks{
Supported
by the Zhongguancun Haihua Institute for Frontier Information Technology.}\\
\small Tsinghua University \\
\and
Shunhua Jiang\thanks{Supported by NSF CAREER award CCF-1844887.}\\
\small Columbia University  
 }

\begin{document}
\date{}
\maketitle
\thispagestyle{empty}
\setcounter{page}{0}

\begin{abstract}
     We consider the energy complexity of the {leader election} problem in the single-hop radio network model, where each device $v$ has a unique identifier $\ID(v) \in \left\{1, 2, \ldots, N\right\}$.
Energy is a scarce resource for small battery-powered devices. For such devices, most of the energy is often spent on communication, not on computation. To approximate the actual energy cost, the energy complexity of an algorithm is defined as the maximum over all devices of the number of time slots where the device transmits or listens. 
      
Much progress has been made in understanding the energy complexity of leader election in radio networks, but very little is known about the trade-off between time and energy.  Chang~et~al.~{[}STOC 2017{]} showed that the optimal deterministic energy complexity of leader election is  $\Theta(\log \log N)$ if each device can simultaneously transmit and listen, but still leaving the problem of determining the optimal time complexity under any given energy constraint.

\begin{description}
     \item[Time-energy trade-off:] 
     For any $k \geq \log \log N$, we show that a leader among at most $n$ devices can be elected deterministically in $O(k \cdot n^{1+\epsilon}) + O(k \cdot N^{1/k})$ time and $O(k)$ energy if each device can simultaneously transmit and listen, where $\epsilon > 0$ is any small constant. This improves upon the previous $O(N)$-time  $O(\log \log N)$-energy algorithm by Chang~et~al.~{[}STOC 2017{]}. We  provide lower bounds to show that the time-energy trade-off of our algorithm is near-optimal.
     \item[Dense instances:] For the dense instances where the number of devices is $n = \Theta(N)$, we design a deterministic leader election algorithm using only $O(1)$ energy.
     This improves upon the $O(\log^\ast N)$-energy algorithm by Jurdzi\'{n}ski, Kuty\l{}owski, and Zatopia\'{n}ski~{[}PODC 2002{]} and the $O(\alpha(N))$-energy algorithm by Chang et al.~{[}STOC 2017{]}.  More specifically, we show that the optimal deterministic energy complexity of leader election is $\Theta\left(\max\left\{1, \log \frac{N}{n}\right\}\right)$ if each device cannot simultaneously transmit and listen, and it is  $\Theta\left(\max\left\{1, \log \log \frac{N}{n}\right\}\right)$ if each device can simultaneously transmit and listen.
\end{description}

\end{abstract}


\newpage

\section{Introduction}\label{sect-intro}

A \emph{radio network} is a distributed system of devices equipped
with radio transceivers. In this paper, we focus on \emph{single-hop} networks where all devices communicate via  a single communication channel. The communication proceeds in synchronous rounds. In each time slot, each device performs some computation and chooses to either 
 listen to the channel,  transmit a message, or stay idle.

The \emph{energy complexity} of a distributed algorithm is the  maximum over all devices of the number of time slots where the device is not idle~\cite{ChangKPWZ17,Chang18broadcast,Chang20bfs,JurdzinskiKZ02podc,JurdziskiKZ03}. That is, each transmitting and listening time slot costs one unit of energy.
This complexity measure is motivated by the fact that energy is a scarce resource for small battery-powered devices, and for these devices a large fraction of energy is often spent on communication, not on computation.

\paragraph{Collision detection.} We consider the following four collision detection models~\cite{ChangKPWZ17}, depending on whether the devices have \emph{sender-side} collision detection (i.e., devices can simultaneously transmit and listen) and \emph{receiver-side} collision detection (i.e., devices can distinguish between collision and silence).

\begin{description}
\item[$\strongcd$.] Transmitters and listeners receive one of the three  feedback:
(i)~silence, if zero devices transmit, 
(ii)~collision, if at least two devices transmit, or
(iii)~a message $m$, if exactly one device transmits.

\item[$\sendercd$.] Transmitters and listeners receive one of the two   feedback:
(i)~silence, if zero or at least two devices transmit, or
(ii)~a message $m$, if exactly one device transmits.

\item[$\receivercd$.] Transmitters receive no feedback.  Listeners receive one of the three  feedback:
(i)~silence, if zero devices transmit, 
(ii)~collision, if at least two devices transmit, or
(iii)~a message $m$, if exactly one device transmits.

\item[$\nocd$.] Transmitters receive no feedback. Listeners receive one of the two   feedback:
(i)~silence, if zero or at least two devices transmit, or
(ii)~a message $m$, if exactly one device transmits.
\end{description}

\paragraph{Deterministic leader election.} Through the paper, a single-hop radio network is described by a set  of devices $V$, where each device $v \in V$ has a distinct identifier $\ID(v) \in [N] = \left\{1, 2, \ldots, N\right\}$. The parameter $N$ is a global knowledge.
The goal of the \emph{leader election} problem is to have exactly one device   identify itself as the \emph{leader} and  all other devices identify themselves as \emph{non-leaders}.

\paragraph{A simple $O(N)$-time and $O(\log N)$-energy algorithm in $\nocd$.}
It is well-known that leader election can be solved in $O(N)$ time and $O(\log N)$ energy in the $\nocd$ model~\cite{nakano2000randomized}, as follows.
The set $S$ initially contains all devices $u$ whose $\ID(u)$ is an odd number.
For $i= 1, 2, \ldots, \ceil{N/2}$, in the $i$th time slot, the device $u$  with $\ID(u) = 2i-1$  transmits a dummy message,  the device $v$ with $\ID(v) = 2i$ listens, and $v$ adds itself to $S$ if $v$ \emph{does not hear} a message from $u$.
It is clear that  if there is \emph{at least one} device $w$ with $\ID(w) \in\left\{2i-1, 2i\right\}$, then \emph{exactly one} such device is in $S$. We can solve the leader election problem by having each $w \in S$ reset its identifier $\ID(w) \leftarrow \ceil{\ID(w) / 2}$ and recursing on the devices in $S$ with the new $\ID$ space $\left\{1, 2, \ldots,\ceil{N/2}\right\}$. 

The time complexity $T(N)$ and the energy complexity $E(N)$ of the algorithm satisfies the following recurrence relations.
\begin{align*}
    T(N) &= T\left(\ceil{N/2}\right) + \ceil{N/2}, & E(N) &= E\left(\ceil{N/2}\right) + 1,\\
    T(1) &= 0, & E(1) &= 0.&
\end{align*}
Hence $T(N) = O(N)$ and $E(N) = O(\log N)$.

This algorithm is \emph{optimal} in both time and energy in the $\nocd$ model, as
Jurdzi\'{n}ski, Kuty\l{}owski, and Zatopia\'{n}ski showed that leader election requires $\Omega(N)$   time~\cite{JurdzinskiKZ02podc} and $\Omega(\log N)$ energy~\cite{JurdziskiKZ03} in  the $\nocd$ model.

\paragraph{Other collision detection models.} For the energy complexity in the $\sendercd$ model, Jurdzi\'{n}ski, Kuty\l{}owski, and Zatopia\'{n}ski proved an  upper bound $O\left(\sqrt{\log N}\right)$~\cite{JurdzinskiKZ02podc} and a lower bound  $\Omega\left(\log \log N / \log \log \log N\right)$~\cite{JurdziskiKZ03}. Later, Chang et al.~\cite{ChangKPWZ17} settled the optimal  energy complexity of leader election in all four models by showing that it is $\Theta(\log N)$ if each device cannot simultaneously transmit and listen (i.e., $\receivercd$ and $\nocd$), and it is  $\Theta(\log \log N)$  if each device can simultaneously transmit and listen (i.e., $\strongcd$ and $\sendercd$).

\paragraph{Time-energy trade-off.}  
To the best of our knowledge, the only existing result relevant to the  trade-off between time and energy is that 
 leader election can be solved in $O(\log N)$ time and  $O(\log N)$ energy in $\receivercd$ and $\strongcd$~\cite{capetanakis1979tree,hayes1978adaptive,tsybakov1978free} by doing a binary search over the ID space $[N]$. For comparison, a much better energy complexity of $O(\log \log N)$ can be achieved in $\strongcd$ with an algorithm that has a worse time complexity of $O(N)$~\cite{ChangKPWZ17}.
 
The algorithm works as follows. In one round of communication, the size of the ID space can be reduced to $\ceil{N/2}$, as follows. Let $L$ be the set of devices $v$ with $\ID(v) \leq \ceil{N/2}$, and let $R$ be the remaining devices. The devices in $L$ transmit a dummy message at the same time, and the devices in $R$ listen. If the devices in $R$ hear silence, then we recurse on $R$, since this means the set $L$ is empty. If the devices in $R$ 
detect collision or receive a message, then we recurse on $L$, since this means the set $L$ is non-empty. A leader is elected after $O(\log N)$ depths of recursion.

This algorithm can be  generalized to obtain a time-energy trade-off in $\strongcd$.
Consider any $k  \geq \log \log N$. Apply the above strategy for $k$ iterations to reduce the size of the ID space to  $N' = O\left(\frac{N}{2^k} \right)$, and then run the $O(N')$-time $O(\log \log N')$-energy leader election algorithm of~\cite{ChangKPWZ17}.
The overall time complexity  is $O\left(k + \frac{N}{2^k} \right)$, and the overall energy complexity is $O(k)$.

\paragraph{Dense instances.} The $\Omega(\log N)$ and $\Omega(\log \log N)$ energy lower bounds of Chang et al.~\cite{ChangKPWZ17} only work for the special case of $|V|=2$. Much more energy-efficient leader election algorithms are known for  \emph{dense} instances. When the number of devices is $|V| = \Theta(N)$, Jurdzi\'{n}ski, Kuty\l{}owski, and {Zatopia\'{n}ski}~\cite{JurdzinskiKZ02podc} showed an $O(\log^\ast N)$-energy algorithm, and the energy complexity was later improved to the inverse Ackermann function $O(\alpha(N))$ by Chang et al.~\cite{ChangKPWZ17}.

We briefly explain the ideas underlying these algorithms. In the above simple $O(N)$-time and $O(\log N)$-energy algorithm, by adding an extra time slot for each $i$ to let the device $v$ with $\ID(v) = 2i$  transmit an acknowledgement to the device $u$ with $\ID(u) = 2i-1$,   we can let each device $v \notin S$ that drops out of the algorithm be \emph{remembered} by some device $u \in S$. In particular, as we have $|V| = \Theta(N)$,  after $t$ levels of recursion, most of the remaining devices remember $\Omega\left(2^t\right)$ devices.
Let $B_v$ denote the set of devices remembered by $v$, and let $v$ drop out of the algorithm if $|B_v| = o(2^t)$ is too small. 
The key idea  is that $v$ can let the members in the group $B_v$  share the energy cost in subsequent iterations. As $|B_v| = \Omega\left(2^t\right)$, we can execute $2^t$ levels of recursion with energy cost $O(1)$ per device. After that, most of the remaining devices form groups of size $\Omega\left(2^{2^t}\right)$, and hence the next $2^{2^t}$ levels of recursion  can be executed with energy cost $O(1)$ per device. Repeating this argument, leader election can be solved in $O(\log^\ast N)$ energy. 

Replacing the  $O(\log N)$-energy algorithm  by the new $O(\log^\ast N)$-energy algorithm in the above approach, an algorithm with energy complexity  $O(\log^{\ast \ast} N)$ is obtained, where $\log^{\ast \ast} N$ is defined by $\log^{\ast \ast} 1 = 0$ and $\log^{\ast \ast} n = 1 + \log^{\ast \ast} \log^\ast n$ for $n > 1$. Similarly, the energy complexity can be improved to $O(\log^{\ast \ast \ast} N)$.   Repeating this argument, a leader election algorithm with inverse Ackermann energy complexity
\[
O(\alpha(N)) = O\left( \min \left\{ i \in \mathbb{N} \ \middle| \ \log^{  \overbrace{\ast \ast  \cdots \ast}^{\text{$i$ times}}} N \leq 1 \right\}\right)
\]
is obtained.
Due to the complex recursive structure of this approach, the algorithms~\cite{ChangKPWZ17,JurdzinskiKZ02podc} based on this approach are naturally very complicated.

\paragraph{Summary.}   See \cref{tab:old} for a summary of existing time and energy bound for deterministic leader election. For more related work, see \cref{sec:related}.  
In \cref{tab:old,tab:tradeoff,tab:dense}, network size refers to the prior knowledge of the devices about the size of the network $|V|$.  For example, $|V|=n$ means that there is a prior knowledge that $n$ is the number of devices, $n \leq |V| \leq N$ means that there is a prior knowledge that $n$ is a lower bound on the number of devices, and $1 \leq |V| \leq N$ means that there is no prior knowledge on $|V|$ at all, except for the trivial lower bound $1$ and the trivial upper bound $N$. 
Note that a lower bound that works with a prior knowledge $|V|=n$ is strictly stronger than one that works with a prior knowledge $n \leq |V| \leq N$.

Lower bounds for stronger models immediately apply to weaker models. For example, the $\Omega(\log \log N)$-energy lower bound applies  not only to $\strongcd$ but also to $\sendercd$. Similarly, upper bounds for weaker models immediately apply to stronger models. For example, the $O(N)$-time and $O(\alpha(N))$-energy algorithm for the dense instances $|V| = \Theta(N)$ works in all four models. 

The $\Omega(n)$-time $\sendercd$ lower bound in \cref{tab:old} was originally stated as an $\Omega(N)$-time
lower bound for algorithms that work for all possible network size $1 \leq |V| \leq N$. Such a lower bound immediately implies an $\Omega(n)$-time lower bound for algorithms that work for the range of network size $1 \leq |V| \leq n$ by restricting the ID space from $[N]$ to $[n]$.

\begin{table*}[!ht]
\fontsize{8}{9}\selectfont 
\centering
\resizebox{\linewidth}{!}{
    \begin{tabular}{|l |l | l | l | l | m{3.8cm} |}
\multicolumn{1}{l}{\bf Model} & 
\multicolumn{1}{l}{\bf Time} & 
\multicolumn{1}{l}{\bf Energy} & 
\multicolumn{1}{l}{\bf Network size} &
\multicolumn{1}{l}{\bf Note} &
\multicolumn{1}{l}{\bf Reference}
\\ \hline
\multirow{2}{*}{$\strongcd$} & $O\left(k + \frac{N}{2^k}\right)$ & $ O(k)$  & $1 \leq |V| \leq N$ &  $k \geq \log \log N$ & \cite{capetanakis1979tree,hayes1978adaptive,tsybakov1978free} + \cite{ChangKPWZ17}
\\ \cline{2-6}
& any & $ \Omega(\log \log N)$  & $|V| = 2$ & & \cite{ChangKPWZ17}
\\ \hline
\multirow{2}{*}{$\sendercd$} & $O(N)$ & $ O(\log \log N)$ & $1 \leq |V| \leq N$ & & \cite{ChangKPWZ17}
\\ \cline{2-6}
& $\Omega(n)$ & any & $1 \leq |V| \leq n$ & & \cite{JurdzinskiKZ02podc}
\\ \hline 
\multirow{2}{*}{$\receivercd$} & $O(\log N)$ & $ O(\log N)$ & $1 \leq |V| \leq N$ & & \cite{capetanakis1979tree,hayes1978adaptive,tsybakov1978free}
\\ \cline{2-6}
& any & $ \Omega(\log N)$ & $|V| = 2$ & & \cite{ChangKPWZ17,JurdziskiKZ03}
\\ \hline
\multirow{2}{*}{$\nocd$} & $O(N)$ & $ O(\log N)$ & $1 \leq |V| \leq N$ & & \cite{nakano2000randomized}
\\ \cline{2-6}
& $O(N)$ & $ O(\alpha(N))$ & $\epsilon N \leq |V| \leq N$ & $\epsilon = \Theta(1)$ & \cite{ChangKPWZ17}
\\ \hline
    \end{tabular}
    }
    \caption{Summary of existing time and energy bounds}
    \label{tab:old}
\end{table*}

\subsection{Our contribution}  

The contribution of this paper is as follows.

\paragraph{Time-energy trade-off.} 
In this paper, we show a near-optimal time-energy trade-off for deterministic leader election in $\strongcd$ and $\sendercd$.
 For any $k \geq \log \log N$, we show that a leader among at most $n$ devices can be elected deterministically in $O(k n^{1+\epsilon}) + O(k N^{1/k})$ time and $O(k)$ energy in $\sendercd$, where $\epsilon > 0$ is an arbitrarily small constant. Our algorithm requires that all devices know the parameters $k$ and $\epsilon$. This improves upon the previous $O(N)$-time  $O(\log \log N)$-energy algorithm in~\cite{ChangKPWZ17}. 
 For the case of $\strongcd$, the time complexity can be improved to $O\left(k N^{1/k} + \min\left\{k n^{1+\epsilon}, \frac{N}{2^k}\right\}\right)$ by combining our new algorithm with the  $O\left(k + \frac{N}{2^k}\right)$-time and $O(k)$-energy algorithm that we discussed above.

We also provide lower bounds to show that the time-energy trade-off of our algorithm is near-optimal.
See \cref{tab:tradeoff} for a summary of the results, where $\epsilon > 0$ can be any arbitrarily small constant. The $\Omega(n)$ lower bound of~\cite{JurdzinskiKZ02podc} is also included in the table for showing the near-optimality of our $O\left(kN^{1/k} + k n^{1+\epsilon}\right)$ upper bound in $\sendercd$.
Note that any $\strongcd$ lower bound applies to all four models. The $\Omega(k N^{1/k})$-time lower bound also applies to algorithms that work for the network sizes $1 \leq |V| \leq n$, as the condition $|V| = 2$ is more restricted than the condition $1 \leq |V| \leq n$. 

\begin{table*}[!ht]
\fontsize{8}{9}\selectfont
\centering
\resizebox{\linewidth}{!}{
    \begin{tabular}{|l |l | l | l | l | l |}
\multicolumn{1}{l}{\bf Model} & 
\multicolumn{1}{l}{\bf Time} & 
\multicolumn{1}{l}{\bf Energy} & 
\multicolumn{1}{l}{\bf Network size} &
\multicolumn{1}{l}{\bf Note} &
\multicolumn{1}{l}{\bf Reference}
\\ \hline
\multirow{3}{*}{$\strongcd$} & $O\left(kN^{1/k} + \min\left\{k n^{1+\epsilon}, \frac{N}{2^k}\right\}\right)$ & $ O(k)$  & $1 \leq |V| \leq n$ &  $k \geq \log \log N$ & \cref{thm:timeub}
\\ \cline{2-6}
  & $\Omega(k N^{1/k})$ & $k$ & $|V|=2$ & & \cref{thm:timelb2}
\\ \cline{2-6}
& $\Omega\left( \min\left\{n, \frac{N}{2^k}\right\}\right)$ & $k$  &  $1 \leq |V| \leq n$ & & \cref{thm:timelb1}
\\ \hline
\multirow{2}{*}{$\sendercd$} & $O\left(kN^{1/k} + k n^{1+\epsilon}\right)$ & $O(k)$ & $1 \leq |V| \leq n$ & $k \geq \log \log N$ & \cref{thm:timeub}
\\ \cline{2-6}
& $\Omega(n)$ & any & $1 \leq |V| \leq n$ & & \cite{JurdzinskiKZ02podc}
\\ \hline 
    \end{tabular}
    }
    \caption{New results for time-energy trade-off}
    \label{tab:tradeoff}
\end{table*}

\paragraph{Dense instances.} 
In this paper, we show that leader election  for dense instances $|V| = \Theta(N)$ can be solved in $O(N)$ time and $O(1)$ energy in $\nocd$, improving upon the $O(\alpha(N))$-energy upper bound given in the prior work~\cite{ChangKPWZ17,JurdzinskiKZ02podc}. 
Moreover, our algorithm can be extended to work for all possible network sizes $1 \leq |V| \leq N$, and no prior knowledge of $|V|$ is required.

We also provide energy lower bounds \emph{matching} our energy upper bounds in all four models. More specifically, we show that for algorithms that work for  all possible network size $1 \leq |V| \leq N$ and require no prior knowledge on $n = |V|$, the  \emph{optimal} 
 deterministic energy complexity of leader election is $\Theta\left(\max\left\{1, \log \frac{N}{n}\right\}\right)$ in $\receivercd$ and $\nocd$, and it is $\Theta\left(\max\left\{1, \log \log \frac{N}{n}\right\}\right)$ in $\strongcd$ and $\sendercd$.
 
 See \cref{tab:dense} for a summary of our results. Note that lower bounds for algorithms that work for a more restricted range of network sizes also apply for algorithms that work for a less restricted range of network sizes. In particular, both of our $\Omega\left(\log \frac{N}{n}\right)$ and $\Omega\left(\log \log \frac{N}{n}\right)$ energy lower bounds apply to the setting where the algorithm  works for  all possible network size $1 \leq |V| \leq N$ and require no prior knowledge on $n = |V|$.

\begin{table*}[!ht]
\fontsize{8}{9}\selectfont
\centering
\resizebox{\linewidth}{!}{
    \begin{tabular}{|l |l | l | l | l | l |}
\multicolumn{1}{l}{\bf Model} & 
\multicolumn{1}{l}{\bf Time} & 
\multicolumn{1}{l}{\bf Energy} & 
\multicolumn{1}{l}{\bf Network size} &
\multicolumn{1}{l}{\bf Note} &
\multicolumn{1}{l}{\bf Reference}
\\ \hline
\multirow{2}{*}{$\{\strongcd, \sendercd\}$} & $O(N)$ & $O\left(\max\left\{1, \log \log \frac{N}{n}\right\}\right)$  & $1 \leq |V| \leq N$ &  $n  = |V|$ & \cref{thm:ub-main}
\\ \cline{2-6}
& any & $\Omega\left(\log \log \frac{N}{n}\right)$  & $n \leq |V| \leq N$ &   & \cref{thm:energy_lb_log_log_N_n}
\\ \hline
\multirow{2}{*}{$\{\receivercd, \nocd\}$} & $O(N)$ & $O\left(\max\left\{1, \log \frac{N}{n}\right\}\right)$ & $1 \leq |V| \leq N$ & $n = |V|$ & \cref{thm:ub-main}
\\ \cline{2-6}
& any & $\Omega\left(\log \frac{N}{n}\right)$ & $|V| = n$  &  $2 \leq n \leq N-1$ & \cref{thm:energy_lb_log_N_n}
\\ \hline
    \end{tabular}
    }
    \caption{New results for dense instances}
    \label{tab:dense}
\end{table*}

\subsection{Related Work}\label{sec:related}

We give a brief overview of some additional related work. In this section $n$ denotes the number of devices in the network.

\paragraph{Randomized setting.}
For the  randomized time complexity of leader election,  Willard proved that $\Theta(\log\log n)$ time is necessary and sufficient
for leader election in the $\strongcd$   model with constant success probability~\cite{Willard86}. More generally, the optimal time complexity is $\Theta(\log\log n + \log f^{-1})$ if the maximum allowed failure probability is $f$~\cite{nakano2002uniform,Newport14,Willard86}.

In the $\sendercd$ model,   leader election can be solved in $O(\log  n \log f^{-1})$ time with success probability $1-f$~\cite{JurdzinskiS02,Newport14}. An $\Omega(\log^2  n)$ time lower bound is known for the case $f = 1/n$~\cite{Newport14}. An $\Omega(\log  n \log f^{-1})$ time lower bound  is known~\cite{Farach-ColtonFM06} for the case of \emph{uniform} algorithms,  in the sense that for each time slot $\tau$, there is a transmitting probability $p_{\tau}$ such that each device transmits in time slot $\tau$ with probability $p_{\tau}$ using fresh randomness independently. 

Chang et al.~\cite{ChangKPWZ17} proved that for randomized polynomial-time leader election with success probability $1 - 1/\poly(n)$, the optimal energy complexity is
$\Theta(\log \log^\ast n)$  in $\{\strongcd$, $\receivercd\}$, and it is $\Theta(\log^\ast n)$  in  $\{\sendercd$, $\nocd\}$. 
There is a tradeoff between time and energy. For instance, in the $\nocd$ model, with $O\left(\log^{2+\epsilon} n\right)$
time leader election can be solved in $O\left(\epsilon^{-1}\log\log\log n\right)$ energy, with success probability $1 - 1/\poly(n)$.  In the setting considered in~\cite{ChangKPWZ17},     devices may consume unbounded energy and never
halt in a failed execution. 

Andriambolamalala 
and Ravelomanana~\cite{Andriambolamalala2020} considered the setting where either $n$ or $\log n$ is known, and they designed randomized $O(1)$-energy and $\log^{O(1)} n$-time algorithms for leader election.  

\paragraph{Multi-hop networks.}
For \emph{multi-hop} radio networks, Bar-Yehuda, Goldreich,  and Itai~\cite{bar1991efficient} showed that leader election can be solved by emulating known leader election algorithms in single-hop networks using a \emph{broadcasting} algorithm as a communication primitive. By emulating Willard's algorithm~\cite{Willard86}, they showed that leader election can be solved in $O(\BC \log n)$ time with success probability $1 - 1/\poly(n)$ and in $O(\BC \log \log n)$ time in expectation, where $\BC$ is the time complexity for broadcasting a single message from multiple sources with success probability $1 - 1/\poly(n)$ in the $\nocd$ model.

The seminal \emph{decay} algorithm of Bar-Yehuda, Goldreich,  and Itai~\cite{bar1992time} solves the broadcasting problem 
in $\BC = O\left(D\log n + \log^2 n\right)$ time, where $D$ is the diameter of the network. 
This bound was later improved to $\BC = O\left(D \log \frac{n}{D} + \log^2 n\right)$~\cite{CR06,KowalskiP05}, which is optimal in view of the  $\Omega(\log^2 n)$ lower bound of Alon et al.~\cite{alon1991lower} and the $\Omega\left(D\log\frac{n}{D}\right)$ lower bound of Kushilevitz and  Mansour~\cite{KushilevitzM98}, but the $\Omega\left(D\log\frac{n}{D}\right)$ lower bound only holds in the setting where \emph{spontaneous transmission} is not allowed, i.e.,  devices that have not yet received a message are forbidden to transmit.

Ghaffari and Haeupler~\cite{GhaffariH13} showed that leader election in a multi-hop radio network can be solved in  $O\left(D \log \frac{n}{D} + \log^3 n\right) \cdot \min\left\{\log \log n, \log \frac{n}{D}\right\}$ time with success probability $1 - 1/\poly(n)$ in the $\nocd$ model, improving the previous bound $O(\BC \log n)$ of~\cite{bar1991efficient}.

Haeupler and Wajc~\cite{haeupler2016faster} showed that the 
lower bound   $\Omega\left(D\log\frac{n}{D}\right)$ of~\cite{KushilevitzM98} can be circumvented when spontaneous transmission is allowed, i.e.,  devices can transmit in any time slot. They showed a broadcasting algorithm with time complexity $O\left( D \frac{\log n \log \log n}{\log D}\right)  + \log^{O(1)} n $ with success probability  $1 - 1/\poly(n)$ in the $\nocd$ model.   Czumaj and  Davies~\cite{CzumajD17} improved this bound to $O\left( D \frac{\log n }{\log D}\right)  + \log^{O(1)} n $, and they showed that leader election can be solved with the same time bound.

Chang et al.~\cite{Chang18broadcast,Chang20bfs} studied the energy complexity in multi-hop radio networks with spontaneous transmission. They showed that broadcasting can be solved in   $O(n \log^2 n \log \Delta)$  time and   $O(\log^2 n \log \Delta)$ energy with success probability $1 - 1/\poly(n)$ in the $\nocd$ model, where $\Delta$ is the maximum degree of the network~\cite{Chang18broadcast}. They also showed that   \emph{breadth first search} can be solved in   $O(D) \cdot n^{o(1)}$ time and   $n^{o(1)}$ energy  with success probability $1 - 1/\poly(n)$ in the $\nocd$ model~\cite{Chang20bfs}. The energy complexity of maximal matching in multi-hop radio networks was recently studied by Dani~et~al.~\cite{dani2021wake}.

Energy complexity has  recently  been studied in the $\LOCAL$ model, where each device can send a separate message to each of its neighbors in each time slot. Chatterjee,   Gmyr,  and Pandurangan~\cite{ChatterjeeGP20} showed that a \emph{maximal independent set} can be computed in $O\left(\log^{3.41} n\right)$ time with $O(1)$ \emph{average} energy cost per device, with success probability $1 - 1/\poly(n)$, in the $\LOCAL$ model. 

\section{Time-Energy Trade-Off}

In this section, we present a leader election algorithm in the $\sendercd$ model that uses $O(k)$ energy and runs in $O(k N^{1/k} + k n^{1+ \epsilon})$ time for any constant $\epsilon >0$, and for any $k \geq \log \log N$. In \cref{thm:timeub}, we assume that $n$, $k$, and $\epsilon$ are global knowledge.

\begin{theorem}\label{thm:timeub}
Let $N$ be the size of the ID space. Suppose that $n \geq |V|$ is a known upper bound on the number of devices. Assume that $k \geq \log \log N$, and $\epsilon >0$ is any constant.
There is a deterministic leader election algorithm in the $\sendercd$ model with time complexity $T = O(kN^{1/k} + k n^{1+\epsilon})$ and energy complexity $E  = O(k)$.
\end{theorem}

Combining \cref{thm:timeub} with the $\strongcd$ $O\left(k + \frac{N}{2^k}\right)$-time and $O(k)$-energy algorithm in \cref{sect-intro}, we also obtain that leader election can be solved in
\begin{align*}
    &\min\left\{ O\left(k + \frac{N}{2^k}\right), \; O\left(kN^{1/k} + k n^{1+\epsilon}\right) \right\} \\
    &= O\left(kN^{1/k} + \min\left\{k n^{1+\epsilon}, \frac{N}{2^k}\right\}\right)
\end{align*}
time and $O(k)$ energy in $\strongcd$.

In \cref{lem:timeub-aux}, we assume that $b$ and $\epsilon$ are global knowledge.

\begin{lemma}\label{lem:timeub-aux}
Let $N$ be the size of the ID space. 
Let $b$ be an integer satisfying $|V| \leq b^{1-\tilde{\epsilon}}$, where $0 < \tilde{\epsilon} < 1$ is any constant.
There is a deterministic leader election algorithm in the $\sendercd$ model with time complexity $T = O(b  \cdot \log_b N )$ and energy complexity $E  = O(\log_b N + \log \log b)$.
\end{lemma}

The proof of \cref{lem:timeub-aux} is deferred to \cref{sect:algo_time_ub}. We first prove  \cref{thm:timeub} using  \cref{lem:timeub-aux}.

\begin{proof}[Proof of \cref{thm:timeub}]
Let $n$, $k$, and $\epsilon$ be the parameters in  \cref{thm:timeub}. We may assume that $k \leq O(\log N)$, because both the time and energy complexities in  \cref{thm:timeub} improve if we reduce $k$ from $\omega(\log N)$ to $\Theta(\log N)$. We divide the analysis into two cases.

We first consider the case of $n \leq (N^{1/k})^{1-\epsilon}$.
By setting $b = \ceil{N^{1/k}}$ and $\tilde{\epsilon} = \epsilon$ in \cref{lem:timeub-aux}, it satisfies that $|V| \leq n \leq (N^{1/k})^{1-\epsilon} \leq b^{1 - \tilde{\epsilon}}$, so we may apply \cref{lem:timeub-aux} with the parameters $b$ and $\tilde{\epsilon}$.
We have $\log_b N = O(k)$, so the time and energy complexities of \cref{lem:timeub-aux}  are
\begin{align*}
    T &= O(b \cdot  \log_b N)  = O(k N^{1/k}),\\
    E  &= O(\log_b N + \log \log b) = O(k).
\end{align*}
In the calculation of the energy complexity, we use the assumption  $k \geq \log \log N$ in  \cref{thm:timeub}, so $\log \log b = O(\log \log N) = O(k)$.

Next, we consider the case of $n \geq (N^{1/k})^{1-\epsilon}$.
We pick $b$ to be the smallest integer such that $n \leq b^{1-\epsilon}$, so $b = \Theta(n^{1+\epsilon}) = \Omega(N^{1/k})$, which implies that $\log_b N = O(k)$. Since
$|V| \leq n \leq b^{1-\epsilon}$,
we may apply \cref{lem:timeub-aux} with the parameters $b$ and $\tilde{\epsilon} = \epsilon$, and it has the following time and energy complexities
\begin{align*}
    T &=  O(b \cdot \log_b N) =  O(n^{1+\epsilon}  \cdot k),\\
    E  &= O(\log_b N + \log \log b) = O(k).\qedhere
\end{align*}
\end{proof}

\subsection{Algorithm}\label{sect:algo_time_ub}
In this section, we prove \cref{lem:timeub-aux} using \cref{alg:time_energy_algorithm}.
We define the following notion of a partition.

\begin{definition}[Partition]
We say that $(S_1, S_2, \ldots, S_{b})$ is a partition of $[N]$ if $\bigcup_{j \in [b]} S_j = [N]$ and any two sets $S_i$ and $S_j$ are disjoint.
\end{definition}

\cref{alg:time_energy_algorithm} crucially uses the fact that there exist $K = O(\tilde{\epsilon}^{-1} \log_b N)$ ``good'' partitions $P^{(i)} = (S_1^{(i)}, S_2^{(i)}, \ldots, S_{b}^{(i)})$, for $i \in [K]$, in the sense that there exists at least one $S_j^{(i)}$ that contains \emph{exactly} one device of $V$, for any choice of $V \subseteq [N]$ of size at most $n$. This lemma is proved using a probabilistic method, and we postpone the proof to  \cref{sect:good_partition}. 

\begin{lemma}[Existence of good partitions]\label{lem:good_partition}
Let $\tilde{\epsilon} \in (0,1)$ be a constant. If $n \leq b^{1 - \tilde{\epsilon}}$, then there exist $K = O(\tilde{\epsilon}^{-1} \log_b N)$ partitions $P^{(1)}, P^{(2)}, \ldots, P^{(K)}$ of $[N]$ into $b$ parts satisfying the following.
\begin{itemize}
\item For each subset $V \subseteq [N]$ of size at most $n$, there exist $(i, j)$ such that $\left|S_{j}^{(i)} \cap V\right| = 1$, where $S_{j}^{(i)}$ is the $j$th part of the $i$th partition $P^{(i)} = \left(S_{1}^{(i)}, S_{2}^{(i)}, \ldots, S_{b}^{(i)}\right)$.
\end{itemize}
\end{lemma}

The algorithm also invokes a deterministic leader election algorithm of  \cite{ChangKPWZ17} as a subroutine on an ID space of size $N' = b$, and it costs  $O(N')$ time and $O(\log \log N')$ energy, see \cref{tab:old}.

\cref{alg:time_energy_algorithm} has $K= O(\tilde{\epsilon}^{-1} \log_b N)$ iterations, and in the $i$th iteration we consider the $i$th partition $P^{(i)}$. For $j = 1,2, \ldots, b$, we let all devices in the $j$th set $S_j^{(i)}$ speak simultaneously (Line~\ref{line:transmit_message}). If a device $v$ is the only speaking device in its set, then it marks itself by setting $s(v) \leftarrow 1$ (Line~\ref{line:set_s_v}). A device can detect this because we are in the $\sendercd$ model. We collect all marked devices and run the deterministic leader election algorithm of~\cite{ChangKPWZ17} among them (Line~\ref{line:run_known_algo}). 

When running the leader election algorithm of~\cite{ChangKPWZ17}, $v \in S_j^{(i)}$ use $j$ as its new identifier, so that the new ID space becomes $[b]$. These identifiers are unique because $v$ is the unique device in its set $S_j^{(i)}$ if $s(v) = 1$.
\cref{alg:time_energy_algorithm} stops whenever the set of marked devices is non-empty and thus a leader is successfully elected by the algorithm of~\cite{ChangKPWZ17} (Line~\ref{line:break}).

\begin{algorithm}[ht]\caption{A leader election algorithm in the $\sendercd$ model with near-optimal time-energy trade-off.}\label{alg:time_energy_algorithm}
\SetAlgoLined
Let $P^{(i)} = (S_1^{(i)}, S_2^{(i)}, \ldots, S_{b}^{(i)})$ for $i \in [K]$ be the $K=O(\tilde{\epsilon}^{-1} \log_b N)$ good partitions given by Lemma~\ref{lem:good_partition}\;
\For{$i= 1, 2, \ldots, K$}
{
All devices $v \in V$ set $s(v) \leftarrow 0$\;
\For{$j= 1, 2, \ldots, b$}
{
All devices in set $S_j^{(i)}$ transmits a message at the same time\;\label{line:transmit_message}
\lIf{$v \in S_j^{(i)}$ hears its message (i.e., not silence or collision)}
{$v$ sets $s(v) \leftarrow 1$\label{line:set_s_v}}
}
All devices $v$ with $s(v) = 1$ run the deterministic leader election algorithm of \cite{ChangKPWZ17}\;\label{line:run_known_algo}
\lIf{a leader is elected in the previous step}
{break}\label{line:break}
}
\end{algorithm}

\paragraph{Analysis.}
\cref{alg:time_energy_algorithm} successfully elects a leader if in some iteration $i \in [K]$, there exists a set $S_j^{(i)}$ such that there is exactly one device $v \in S_j^{(i)}$. This is guaranteed by the property of the good partitions of Lemma~\ref{lem:good_partition}.

\cref{alg:time_energy_algorithm} has $K=O(\tilde{\epsilon}^{-1} \log_b N) = O(\log_b N)$ iterations, and in each iteration it takes $O(1)$ energy and $O(b)$ time for the devices to check if they are the only device in their sets. In total this part takes $O(\log_b N)$ energy and $O(b  \cdot \log_b N)$ time. 

Furthermore, the deterministic algorithm of~\cite{ChangKPWZ17} is called at most $K=O(\tilde{\epsilon}^{-1} \log_b N) = O(\log_b N)$  times, but each device only participates in at most one of them, because we stop whenever a leader is elected, and a leader is guaranteed to be elected if the number of participants is at least one.
The algorithm of~\cite{ChangKPWZ17} is executed on an ID space of size $N' = b$, and so it takes $O(\log \log N') = O(\log \log b)$ energy and $O(N') = O(b)$ time. Hence the overall complexities for running the algorithm of~\cite{ChangKPWZ17}  is  $O(Kb) = O(b  \cdot \log_b N)$ time and  $O(\log \log b)$ energy.

To summarize, the overall time complexity is $T = O(b  \cdot \log_b N)$ and the overall energy complexity is $E = O(\log_b N + \log \log b)$.

\subsection{Existence of good partitions}\label{sect:good_partition}
Next we prove the existence of $K$ good partitions in Lemma~\ref{lem:good_partition}. The proof relies on a balls-into-bins tool proved in Lemma~\ref{lem:balls_in_bins}.

\begin{lemma}[Balls into bins]\label{lem:balls_in_bins}
Let $n$ and $b$ be two integers that satisfy $n \leq b/2$. There are $n$ balls and $b$ bins. Each time a ball is uniformly randomly thrown into a bin. Then with probability at least $1 - (\frac{n}{b})^{\Omega(n)}$, there exists at least one bin that contains exactly one ball.
\end{lemma}
\ifdefined\isspaa
We postpone the proof to \cref{sec:pf_balls_in_bins}.
\else
\begin{proof}
We define $n$ random variables $X_1, X_2, \ldots, X_n \in \{0,1\}$ as follows. Before throwing the $i$th ball into a random bin, we first re-order the bins such that the non-empty bins have the smallest indices. We define $X_i = 1$ if the $i$th ball is thrown into a bin whose index is in $\{n+1, \ldots, b\}$, and otherwise we define $X_i = 0$.

It is easy to see that $X_1, X_2,\ldots,X_n$ are i.i.d.~random variables with probability $p = 1 - \frac{n}{b}$ to be 1. Since there can be at most $i-1 < n$ non-empty bins before we throw the $i$th ball, we know that the $i$th ball is thrown into an empty bin if $X_i = 1$. Hence there must exist a bin that contains exactly one ball if $\sum_{i=1}^n X_i > \frac{n}{2}$.

Next, we upper bound the probability of the bad event that $\sum_{i=1}^n X_i \leq \frac{n}{2}$. Define $\delta = \frac{1}{2} - \frac{n}{b}$. Using the Chernoff bound, we have
\begingroup
\allowdisplaybreaks
\begin{align*}
    \Pr\left[\sum_{i=1}^n X_i \leq \frac{n}{2}\right] 
    &= \Pr\left[\frac{1}{n} \sum_{i=1}^n X_i \leq p - \delta\right] \\
    &\leq \left( \left( \frac{p}{p - \delta} \right)^{p - \delta} \cdot \left( \frac{1 - p}{1 - p + \delta} \right)^{1 - p + \delta} \right)^{n} \\
    &\leq \left( 2^{1/2} \cdot \left( \frac{2n}{b} \right)^{1/2} \right)^{n}\\
    &= \left(\frac{4n}{b}\right)^{n/2}.
\end{align*}
\endgroup
Thus, with probability at least $1 - (\frac{4n}{b} )^{n/2}$, we have $\sum_{i=1}^n X_i > \frac{n}{2}$, which implies that there exists at least one bin that contains exactly one ball.
\end{proof}
\fi

\begin{proof}[Proof of \cref{lem:good_partition}]
Consider a fixed set $V \subseteq [N]$ of size $\tilde{n}$, where $1 \leq \tilde{n} \leq n$. We independently generate $K = \Theta(\tilde{\epsilon}^{-1} \log_b N)$ random partitions such that each partition consists of $b$ sets, and each number $x \in [N]$ is uniformly randomly put into one set in the partition. We have
\begin{align*}
    \Pr\left[\exists S^{(i)}_j \text{~s.t.~} |V \cap S^{(i)}_j| = 1\right] 
    &= 1 - \prod_{i=1}^{K} \Pr\left[\forall j \in [b], |V \cap S^{(i)}_j| \neq 1\right] \\
    &\geq 1 - \left(\frac{\tilde{n}}{b}\right)^{\Omega(\tilde{n}) \cdot K}\\
    &\geq  1 - b^{-\Omega(\tilde{n}) \cdot K \cdot \tilde{\epsilon}}\\
    &\geq 1 - N^{-\Omega(\tilde{n})}.
\end{align*}

The first step follows from the fact that the $K$ partitions are independent.
The second step follows from Lemma~\ref{lem:balls_in_bins} with $\tilde{n}$ balls and $b$ bins.
The third step follows from $\tilde{n} \leq n \leq b^{1-\tilde{\epsilon}}$, which implies $\frac{\tilde{n}}{b} \leq b^{-\tilde{\epsilon}}$.
The fourth step follows from $K = \Theta(\tilde{\epsilon}^{-1} \log_b N)$.

We write $B_{V}$ to denote the bad event that $|V \cap S^{(i)}_j| \neq 1$ for all $S^{(i)}_j$, and we write $B_{\tilde{n}}$ to denote the bad event that $B_{V}$ occurs for at least one subset $V \subseteq [N]$ of size $\tilde{n}$.  
The above calculation implies that for any given large constant $C$, we can select $K = \Theta(\tilde{\epsilon}^{-1} \log_b N)$ to be sufficiently large to make $\Pr[B_{V}] \leq N^{-C \tilde{n}}$, where $\tilde{n}$ is the size of $V$.

There are in total $\binom{N}{\tilde{n}}   \leq N^{\tilde{n}}$ number of subsets $V \subseteq [N]$ of size $\tilde{n}$. Using a union bound, we know that $B_{\tilde{n}}$ occurs with probability  at most $N^{-(C-1)\tilde{n}}$.
With another union bound over $1 \leq \tilde{n} \leq n$, the probability that there exist $(i, j)$ such that $\left|S_{j}^{(i)} \cap V\right| = 1$ for each subset $V \subseteq [N]$ of size at most $n$ is at least $1 - \sum_{\tilde{n} = 1}^n N^{-(C-1)\tilde{n}} > 0$ by selecting $C$ to be a sufficiently large constant. Since the probability is non-zero, there must exist $K$ partitions that satisfy the desired property.
\end{proof}

\section{Dense Instances}\label{sect-ub}

The goal of this section is to prove the following theorem.

\begin{theorem}
\label{thm:ub-main}
Let $N$ be the size of the ID space. There is a deterministic leader election algorithm with time complexity $T = O(N)$ and energy complexity
\[
\small
E = \begin{cases}
O\left(\max\left\{1, \log  \frac{N}{n}\right\}\right)  &\text{for \ $\receivercd$, $\nocd$,}\\[0.2cm] 
O\left(\max\left\{1,\log \log \frac{N}{n}\right\}\right)  &\text{for \  $\strongcd$, $\sendercd$.}
\end{cases}
\]
The algorithm does not require any prior knowledge of the number of devices $n = |V|$.
\end{theorem}

The algorithm for \cref{thm:ub-main} works for all possible network sizes $1 \leq |V| \leq N$, and it does not require any prior knowledge on $|V|$. Throughout this section, we write $n = |V|$ to denote the unknown number of devices.

\begin{definition}[Group]\label{def:group}
A sequence of devices $S = (v_1, v_2, \ldots, v_s)$ forms a group if each $v_{\ell} \in S$ has $\rank(v_{\ell}) = \ell$ and each $u \in V \setminus S$ has $\rank(u) = \bot$. 
\end{definition}

In \cref{def:group}, $\rank(v)$ is a variable stored in the device $v$.
 The task of leader election can be reduced to forming a group  $S$ of size at least one, as we can let the unique device $v$ with $\rank(v) = 1$  be the leader. Note that the requirement of $\rank(u) = \bot$ for each  $u \in V \setminus S$  in \cref{def:group} implies that each device $u \in V \setminus S$ is aware of the fact that it does not belong to the group $S$.

\subsection{A Simple Algorithm}\label{sect-simple-alg}

In this section, we show a very simple algorithm that elects a leader in $O(N)$ time and $O(1)$ energy when $n = \Theta(N)$.
 We will show that choosing $b$ as any integer such that $n > \ceil{N/b}$, the output of \cref{alg:simple} is a group of at least one device. \cref{alg:simple} solves the leader election problem as we can let the unique device $v$ with $\rank(v) = 1$ be the leader.

\paragraph{High-level idea.} The main idea behind our algorithm is that we try to maintain and grow a list of devices $S = (v_1, v_2, \ldots, v_x)$. If each $v_{\ell}$ knows its rank $\ell$ and each $v \in V \setminus S$ knows that it is not in $S$, then $S$ forms a group. As we will later see, the list $S$ might become empty for several times during the algorithm, but it is guaranteed to be non-empty at the end of the algorithm.

There will be $\ceil{N/b}$ iterations. During the $i$th iteration of the algorithm, the first device $v_1$ in the current list $S$ will be responsible for recruiting the next batch of devices in the ID space $\{b(i-1)+1, b(i-1)+2, \ldots, \min\{N, bi\}\}$ to join $S$. If $S$ is empty at the beginning of an iteration, then the first device in the batch considered in this iteration will become the first member $v_1$ of $S$.

After that, $v_1$ will drop out from $S$. This step is necessary because $S$ might become empty in some subsequent iteration, and $v_1$ will not be able to know that without spending energy.
Because the number of iterations $\ceil{N/b}$ is less than the number of devices $n$, the list $S$ is guaranteed to be non-empty at the end of the algorithm.

Each device $v$ will spend $O(1)$ energy when it is recruited into $S$. 
When $n = \Theta(N)$, the size of the ID space $\{b(i-1)+1, b(i-1)+2, \ldots, \min\{N, bi\}\}$ considered in the $i$th iteration is at most a constant $b = O(1)$, and so the energy cost for the device responsible for recruiting new devices is also $O(1)$. Hence the energy complexity of our algorithm is $O(1)$.

\begin{algorithm}[!ht]
\SetAlgoLined
All devices $w \in V$ initialize $r(w) \leftarrow \bot$ and $s(w) \leftarrow \bot$\;
\label{ll1}\For{$i= 1, 2, \ldots, \ceil{N/b}$}
{
 \For{$j= b(i-1)+1, b(i-1)+2, \ldots, \min\{N, bi\}$\label{ll2}}
{
Let $u$ be the device with $r(u) = i$\;\label{ll21}
Let $v$ be the device with $\ID(v) = j$\;\label{ll22}
\label{ll3} $u$ transmits $s(u)$ and $v$ listens\;
\label{ll4} $v$ transmits a dummy message and $u$ listens\;
  \leIf{$v$ hears a message from $u$}{$v$ sets  $r(v) \leftarrow s(u)+1$\;\label{ll6}
   }{$v$ sets $r(v) \leftarrow i$ and $s(v) \leftarrow i$\label{ll7}}
  \lIf{$u$ hears a message from $v$}
  {$u$ sets $s(u) \leftarrow s(u)+1$\label{ll5}}
 }\label{ll8}
Let $u'$ be the device with $r(u') = i$\;\label{ll9}
Let $v'$ be the device with $r(v') = i+1$\;\label{ll91}
$u'$ transmits $s(u')$ and $v'$ listens\;
$v'$ sets $s(v') \leftarrow s(u')$\;\label{ll10}
}
\ForEach{$w \in V$}
{
\leIf{$r(w) \geq \ceil{N/b}+1$ }
{$\rank(w) \leftarrow r(w) - \ceil{N/b}$\;}
{$\rank(w) \leftarrow \bot$}
}
 \caption{A simple leader election algorithm.}\label{alg:simple}
\end{algorithm}

\paragraph{Maintenance of a group.}
 \cref{alg:simple} is a realization of the above high-level idea. During the algorithm, each device $w \in V$ will maintain two variables $r(w)$ and $s(w)$.  
 At the beginning of the $i$th iteration, we implicitly maintain a group   $S=(v_1, v_2, \ldots, v_x)$ by setting 
\[
\rank(w) = 
\begin{cases}
r(w) - i + 1 & \text{if $r(w) \geq i$,}\\
\bot & \text{otherwise,}
\end{cases}
\]
and we always have: (Note that $x=|S|$.)
\[s(v_1) = x + i -1.\]
In particular, we have the following observations.
\begin{itemize}
    \item The first member $v_1$ of the group $S$ can learn the size $x$ of $S$ by reading $s(v_1)$.
    \item Each $w \in S$ can learn its rank in the group  by reading $r(w)$.
    \item Each $w \in V \setminus S$ can learn the fact that $w$ is not in the group  by reading $r(w)$.
\end{itemize}

That is, the variable $r(w)$ allows each  $w \in V$ to calculate $\rank(w)$ w.r.t.~$S$, and $s(v_1)$ allows the first member $v_1$ of the group $S$ to calculate the size of $S$. 
The knowledge about the size of $S$ is crucial for  $v_1$ because it needs this information to inform each new member of $S$ its rank in $S$.

In \cref{ll21,ll22,ll9,ll91} of  \cref{alg:simple}, it is fine if no such device exists. For example, if $u$ does not exist and $v$ exists, then $v$ will not hear a message from $u$, so $v$ will set $r(v) \leftarrow i$ and $s(v) \leftarrow i$ according to \cref{ll7}. This occurs when the group $S$ at the beginning of the $i$th iteration is empty and the device $v$ with $\ID(v) = j$ is the first device  among all devices whose $\ID$ is in the range $\{b(i-1)+1, b(i-1)+2, \ldots, \min\{N, bi\}\}$. After setting $r(v) \leftarrow i$ and $s(v) \leftarrow i$, $v$ becomes the first member $v_1$ of $S=(v_1)$.

\paragraph{Algorithm.}
In  \cref{alg:simple}, 
initially we have $S = \emptyset$.  For  $i= 1, 2, \ldots, \ceil{N/b}$,   we let the first member $u = v_1$ of the group $S$  recruit new members $v$ to join the group (\Crefrange{ll2}{ll8}). Specifically, device $u$ will communicate with each device $v$ with $\ID(v) = j$, for each $j = b(i-1)+1, b(i-1)+2, \ldots, \min\{N, bi\}$ one by one. 

After finishing this task, device $u$ drops out of the group $S$ when we proceed to the next iteration $i+1$, as $r(u) = i$. The reason that $u$ needs to drop out of the group $S$ is as follows. 
If there is a long interval $I$ of $\ID$s     not held by any device, then the existing group members cannot afford the energy to continue recruiting during the steps $j \in I$. Letting old members to drop out of the group resolves the issue, as no one needs to spend energy if the group is empty and no one attempts to join the group.


For the case where the device $v$ with $\ID(v) = j$ exists and $S \neq \emptyset$,
 to add $v$ to the current group $S =(v_1, v_2, \ldots, v_x)$,   device $u$ sends the number $s(u) = x + i - 1$ to $v$, and then $v$ appends itself to the end of the list $(v_1, v_2, \ldots, v_x)$ by setting $r(v) \leftarrow s(u)+1$ (\cref{ll6}). After that, we have  $S =(v_1, v_2, \ldots, v_{x+1})$ with $v = v_{x+1}$. To reflect this change, the device $u$ updates $s(u) \leftarrow s(u)+1$ (\cref{ll5}), so that $s(u) = (x+1) + i - 1$ reflects the new size $x+1$ of the group $S$.

For the case the device $v$ with $\ID(v) = j$ exists and $S = \emptyset$, device $v$  becomes the first member $v_1$ of the group $S=(v_1)$ by setting $r(v) \leftarrow i$ and $s(v) \leftarrow i$ (\cref{ll7}).

After the first member $u' = v_1$ of the group $S$ finishes its job of recruiting new members, it informs the second member $v' = v_2$ of the current size of the group by sending the number $s(u')$ to $v'$ (\Crefrange{ll9}{ll10}).

\paragraph{Analysis.} 
\cref{alg:simple} elects a leader if the final group size is at least one. It is clear that the size of the final group is at least $n - \ceil{N/b}$, because the number of devices that drop out of the group is at most $\ceil{N/b}$, in view of the above discussion. More formally, each device $w \in V$ is assigned a distinct positive integer $r(w)$ during the iteration $j = \ID(w)$, and $w$ is included in the final group if $r(w) \geq \ceil{N/b} + 1$. Hence the size of the final group is at least $n - \ceil{N/b}$.

We can set $b = O(N/n)$
to satisfy $n - \ceil{N/b} > 0$.
It is clear that the runtime of \cref{alg:simple} is $T = O(N)$, and the energy cost per device is $E= O(b) = O\left(c^{-1}\right)$, where $c = n/N$.

\paragraph{Exponential search.} If the density $c = n/N$ is unknown, then we can do an exponential search for each $b = 2, 4, 8, \ldots$ until \cref{alg:simple} returns a group of size at least one. To test if the group has size at least one, we can let the device $w$ with $\rank(w) = 1$ transmit and let all other devices listen. As long as $n = \Theta(N)$, the exponential search finishes within a constant number of iterations, and so leader election can be solved with time complexity $T = O(N)$ and energy complexity $E = O(1)$ in the deterministic $\nocd$ model.

\subsection{An Improved Algorithm}

Our simple algorithm in \cref{sect-simple-alg} has energy complexity $O\left(c^{-1}\right)$, where $c = n/N$. We show that by combining our approach in \cref{sect-simple-alg} with   existing  algorithms in~\cite{ChangKPWZ17}, the energy complexity can be further improved to $O\left(\max\left\{1, \log c^{-1}\right\}\right)$ for the case of $\receivercd$ and $\nocd$ and  $O\left(\max\left\{1, \log \log c^{-1}\right\}\right)$ for the case of  $\strongcd$ and $\sendercd$.

 \begin{algorithm}[!ht]
\SetAlgoLined
All devices $w \in V$ initialize $r(w) \leftarrow \bot$ and $s(w) \leftarrow \bot$\;

\label{lll1}\For{$i= 1, 2, \ldots, \ceil{N/b}$}
{
Let $V_i = \{ v \in V \ | \ b(i-1)+1 \leq \ID(v) \leq \min\{N, bi\} \}$ and $s_i = |V_i|$\;
 Arrange the set of devices $V_i$ in some order $V_i = (v_{i,1}, v_{i,2}, \ldots, v_{i, s_i})$\; \label{lll2}
 \tcc{It is required that  $v_{i,j}$ knows its rank $j$ and the size $s_i$ of $V_i$.}
Let $u$ be the device with $r(u) = i$\;\label{lll3}
Let $v = v_{i,1}$\;
 $u$ transmits $s(u)$ and $v$ listens\;
$v$ transmits $s_i$ and $u$ listens\; \label{lll4} 
  \leIf{$v$ hears a message from $u$}{$v$ sets  $r(v) \leftarrow s(u)+1$\;\label{lll6}
   }{$v$ sets $r(v) \leftarrow i$ and $s(v) \leftarrow i$\label{lll7}}
  \lIf{$u$ hears a message from $v$}
  {$u$ sets $s(u) \leftarrow s(u)+s_i$} \label{lll5}
  \For{$j= 1, 2, \ldots, s_i - 1$ \label{lll22}}
{$v_{i,j}$ transmits $r(v_{i,j})$ and $v_{i,j+1}$ listens\;
  $v_{i, j+1}$ sets $r(v_{i,j+1}) \leftarrow r(v_{i,j}) + 1$. 
 }\label{lll8}
Let $u'$ be the device with $r(u') = i$\;\label{lll9}
Let $v'$ be the device with $r(v') = i+1$\;
$u'$ transmits $s(u')$ and $v'$ listens\;
$v'$ sets $s(v') \leftarrow s(u')$\;\label{lll10}
}
\ForEach{$w \in V$}
{\leIf{$r(w) \geq \ceil{N/b}+1$}
{$\rank(w) \leftarrow r(w) - \ceil{N/b}$\;}
{$\rank(w) \leftarrow \bot$}
}
 \caption{An improved leader election algorithm.}\label{alg:improved}
\end{algorithm}

\paragraph{Algorithm.}  \cref{alg:improved} is the result of applying the following modifications to \cref{alg:simple}.

Let us write $V_i = \{ v \in V \ | \ b(i-1)+1 \leq \ID(v) \leq \min\{N, bi\} \}.$
The main source of inefficiency of \cref{alg:simple} is that when the device $u$ with $r(u) = i$ recruits the devices  $v \in V_i$ to join the group $S$, it needs to go through the identifiers $j= b(i-1)+1, b(i-1)+2, \ldots, \min\{N, bi\}$ one by one, and this costs $O(b) = O\left(c^{-1}\right)$ energy. 

This energy cost can be reduced to $O(1)$ if the devices in $V_i$ have been arranged in some order $V_i = (v_{i,1}, v_{i,2}, \ldots, v_{i, s_i})$, and each $v_{i,j}$ knows its rank $j$ and the size $s_i$ of $V_i$ (\cref{lll2}).

Specifically, if such an ordering is given, then $u$ only needs to communicate with $v = v_{i,1}$ (\Crefrange{lll3}{lll5}). After that, we can add all devices $v \in V_i$ to the group $S$ by having $v_{i,j}$ communicate with $v_{i,j+1}$ for each $j = 1, 2, \ldots, s_i - 1$ (\crefrange{lll22}{lll8}).

\paragraph{Analysis.}
Similar to the analysis in \cref{sect-simple-alg}, as long as $n - \ceil{N/b} > 0$, \cref{alg:improved} returns a group of size at least one, and so a leader is elected. Hence setting $b = 
 O(N/n) = O\left(c^{-1}\right)$ works.

Due to the modifications, the  energy   complexity of \cref{alg:improved} is $O(1)$, except the part of arranging the devices $V_i$ in some order $V_i = (v_{i,1}, v_{i,2}, \ldots, v_{i, s_i})$ (\cref{lll2}).
For each $i = 1, 2, \ldots, \ceil{N/b}$, the task of this part can be solved using the \emph{census} algorithms of~\cite{ChangKPWZ17}, where  some device announces a list of the IDs of all devices at the end of the computation.
The time complexity of the census algorithms of~\cite{ChangKPWZ17} is $O(b)$. The energy complexity is $O(\log b)$ in $\{\receivercd, \nocd\}$~\cite[Lemma 15]{ChangKPWZ17}, and it is $O\left(\log \log b\right)$ in $\{\strongcd, \sendercd\}$~\cite[Theorem 5]{ChangKPWZ17}.

Hence the time complexity of \cref{alg:improved} is $T = O(N)$, and the energy complexity  of \cref{alg:improved} is $E = O\left(\max\left\{1, \log c^{-1}\right\}\right)$ in $\{\receivercd$, $\nocd\}$ and   is $E = O\left(\max\left\{1, \log \log c^{-1}\right\}\right)$ in $\{\strongcd$, $\sendercd\}$.

\paragraph{Exponential search.} We have proved \cref{thm:ub-main} for the case $n$ is known. To extend our algorithm to the case when $n$ is unknown, we use an exponential search, as in \cref{sect-simple-alg}. For the case of $\{\receivercd, \nocd\}$, to ensure that the overall energy complexity is still $O\left(\max\left\{1, \log c^{-1}\right\}\right)$, in the $i$th attempt of \cref{alg:improved}, we use 
\begin{align*}
   b &= \min\left\{N, 2^{2^{i}}\right\}. \\
   \intertext{Similarly, for the case of $\{\strongcd, \sendercd\}$, to ensure that the overall energy complexity is still $O\left(\max\left\{1, \log \log c^{-1}\right\}\right)$, in the $i$th attempt of \cref{alg:improved}, we use }
   b &= \min\left\{N, 2^{2^{2^i}}\right\}.
\end{align*}

As long as $2 \leq n \leq N$, the number of attempts until a leader is elected is $O\left(\max\left\{1, \log \log c^{-1}\right\}\right)$ for the case of $\{\receivercd$, $\nocd\}$, and it is $O\left(\max\left\{1, \log \log \log c^{-1}\right\}\right)$ for the case of $\{\strongcd$, $\sendercd\}$.
If no leader is elected in all attempts, then we know that the number of devices is one, and so the unique device in the network can elect itself as the leader.
The overall time complexity is  $\omega(N)$ if $c = o(1)$. 
To reduce the time complexity back to $O(N)$, we consider the following trick.

Recall the  procedure described in \cref{sect-intro} that reduces the $\ID$ space from $\{1, 2, \ldots, N\}$ to $\{1, 2, \ldots, \ceil{N/2}\}$ in time $O(N)$  and  energy $O(1)$.  The set $S$ initially contains all devices $u$ whose $\ID(u)$ is an odd number. For $i = 1,2, \ldots, \ceil{N/2}$,  the device $u$  with $\ID(u) = 2i-1$  transmits a dummy message,  the device $v$ with $\ID(v) = 2i$ listens, and $v$ adds itself to $S$ if $v$ {does not hear} a message from $u$.
It is clear that  if there is at least one device $w$ with $\ID(w) \in\{2i-1, 2i\}$, then exactly one such device is in $S$. 
The $\ID$ space is reduced to  $\{1, 2, \ldots, \ceil{N/2}\}$ by having
 each $w \in S$ reset their identifier $\ID(w) \leftarrow \ceil{\ID(w) / 2}$. 
 
 If a device $w \in V$ is not in $S$, then $w$ drops out of the network.
 The density of the new instance 
 $S$ is at least the density of the old instance $V$. That is, we have $n' / N' \geq n/ N$, where $n' = |S|$ and $N' = \ceil{N/2}$.

 When we do the exponential search, after each attempt of \cref{alg:improved}, we run the above procedure to reduce the $\ID$ space. Hence the time complexity of the $i$th attempt becomes $O(N/2^{i-1})$, and  the overall time complexity becomes $O(N)$, as required in \cref{thm:ub-main}.
\section{Lower Bounds}
In this section we prove two energy lower bounds and two time lower bounds for leader election algorithms. We emphasize that these lower bounds work with different prior knowledge about the network size $|V|$. For example, the $\Omega(\log \frac{N}{n})$ energy lower bound of \cref{thm:energy_lb_log_N_n} works with the setting where the number of devices $n = |V|$ is a global knowledge, and the $\Omega(\min\{n, \frac{N}{2^k}\})$ time lower bound of \cref{thm:timelb1} works with the setting where an upper bound $n \geq |V|$ on the number of devices is a global knowledge.

In our lower bound proofs, we consider an easier success criterion. We say a leader election algorithm is successful in some time slot if a message is successfully transmitted, i.e., exactly one device transmits and at least one device listens in that time slot. 

For any leader election algorithm $\mathcal{A}$, we can always add an extra time slot in the end and let the leader transmit and let all other devices listen. Hence we can work with the above easier  success criterion when proving lower bounds.

\subsection{An  \texorpdfstring{$\Omega(\log \frac{N}{n})$}{Omega(log N/n)} Energy Lower Bound for  \texorpdfstring{$\receivercd$ and $\nocd$}{Receiver-CD and No-CD}}
In this section we prove an energy lower bound of $\Omega(\log \frac{N}{n})$ in the $\receivercd$ model. This lower bound also holds in the weaker $\nocd$ model.

\begin{theorem}
[Energy lower bound of $\Omega(\log \frac{N}{n})$]\label{thm:energy_lb_log_N_n}
Let $N$ be the size of the ID space. Let $\mathcal{A}$ be a deterministic leader election algorithm in the $\receivercd$ model that works for any set of devices $V$ of size $n$. Let $k$ denote the energy complexity of $\mathcal{A}$. Then we have
\begin{align*}
    k \geq \Omega(\log \frac{N}{n}).
\end{align*}
\end{theorem}
\begin{proof}
Let $t$ denote the time complexity of $\mathcal{A}$.
For any $j \in [N]$ and $i \in [t]$, define $a_i(j) \in \{\idle, \listen, \transmit\}$ as the action of the device $v$ with $\ID(v) = j$ in the $i$th time slot if in the first $i-1$ time slots whenever $v$ listens it hears silence. Since we are in the $\receivercd$ model, whenever $v$ transmits it receives no feedback.

Consider a random sequence $\{b_i\}_{i=1}^t$ where each $b_i$ is uniformly randomly sampled from $\{\listen, \transmit\}$. We say $\{b_i\}_{i=1}^t$ matches a sequence $\{a_i(j)\}_{i=1}^t$ if for any $i$, either $a_i(j) = \idle$, or $a_i(j) = b_i$. Since there are at most $k$ $\listen$ or $\transmit$ actions in the sequence $\{a_i(j)\}_{i=1}^t$, it is easy to see that
\[
\Pr[\{b_i\}_{i=1}^t \text{~matches~} \{a_i(j)\}_{i=1}^t] = \frac{1}{2^k}.
\]
Thus in expectation $\{b_i\}_{i=1}^t$ matches $\frac{N}{2^k}$ number of action sequences. This means there must exist some $\{b_i\}_{i=1}^t$ that matches at least $\frac{N}{2^k}$ number of action sequences. Let $V$ denote the set of devices that this $\{b_i\}_{i=1}^t$ matches with, and we have $|V| \geq \frac{N}{2^k}$.

The algorithm $\mathcal{A}$ cannot be correct on any set $V' \subseteq V$, because in any time slot all the devices in $V'$ either all perform actions in $\{\listen, \idle\}$, or all perform actions in $\{\transmit, \idle\}$, so there does not exist a time slot where exactly one device transmits and at least one device listens. Since the algorithm $\mathcal{A}$ is correct on all sets of size $n$, we must have $n > |V| \geq \frac{N}{2^k}$, which gives $k \geq \Omega(\log \frac{N}{n})$.
\end{proof}

\subsection{An  \texorpdfstring{$\Omega(k N^{1/k})$}{Omega(k N to the power of 1/k)} Time Lower Bound for All Models}

In this section we prove a time lower bound of $\Omega(k N^{1/k})$ when the energy complexity is $k$. We prove this lower bound in the strongest $\strongcd$ model, so it automatically applies to other models.


Since our lower bound works for the case of exactly two devices, it also applies to more general settings such as $|V| \leq n$ (an upper bound $n$ on the number of devices is given) and  $1 \leq |V| \leq N$ (no knowledge on the number of devices).

\begin{theorem}[Time lower bound of $\Omega(k N^{1/k})$]\label{thm:timelb2}
Let $N$ be the size of the ID space. Let $\mathcal{A}$ be a deterministic leader election algorithm in the $\strongcd$ model that works for any set of devices $V$ with $|V| = 2$. Let $k$ denote the energy complexity of $\mathcal{A}$, and let $t$ denote the time complexity of $\mathcal{A}$. Then we have
\begin{align*}
    t \geq \Omega(k N^{1/k}).
\end{align*}
\end{theorem}
\begin{proof}

For any $j \in [N]$ and $i \in [t]$, define $a_i(j) \in \{\idle$, $\listen$, $\transmit\}$ as the action of the device $v$ with $\ID(v) = j$ in the $i$th time slot if in the first $i-1$ time slots whenever $v$ listens it hears silence, and whenever $v$ transmits it detects collision.
Suppose there exist two different IDs $j$ and $j'$ such that $a_i(j) = a_i(j')$ for all $i \in [t]$. Let $v$ and $v'$ be the devices with $\ID(v) = j$ and $\ID(v') = j'$. The algorithm $\mathcal{A}$ cannot be correct on the set $V = \{v, v'\}$ because $v$ and $v'$ will follow the same action sequence and always transmit or listen at the same time. 
Thus each ID $j \in [N]$ must correspond to a unique action sequence $\{a_i(j)\}_{i=1}^t$. 

Since there are at most $k$ $\listen$ or $\transmit$ actions in the sequence of length $t$, there are in total $\sum_{i=1}^k \binom{t}{i} \cdot 2^i$ possible sequences. When $k\leq t/2$, we have
\begin{align*}
N \leq \sum_{i=1}^k \binom{t}{i} \cdot 2^i
\leq \sum_{i=1}^k \binom{t}{k} \cdot 2^i
\leq \binom{t}{k} \cdot 2^{k+1}
\leq \frac{2(2et)^k}{k^k}.
\end{align*}
Therefore we get $t \geq \Omega(k N^{\frac{1}{k}})$. (When $k=\Theta(t)$, $k=\Omega(\log N)$ and the bound is trivial.)
\end{proof}

\subsection{An  \texorpdfstring{$\Omega(\min\{n, \frac{N}{2^k}\})$}{Omega(min{n, N/(2 to the power of k)})} Time Lower Bound for  \texorpdfstring{$\strongcd$ and $\receivercd$}{Strong-CD and Receiver-CD}}
In this section we prove a time lower bound of $\Omega(\min\{n, \frac{N}{2^k}\})$ for $\strongcd$ and $\receivercd$, generalizing the $\Omega(n)$ lower bound for $\sendercd$ and $\nocd$ shown in \cite{JurdzinskiKZ02podc}. It suffices to prove the new lower bound in  
$\strongcd$. Our proof combines the proof idea of the $\Omega(n)$ lower bound in \cite{JurdzinskiKZ02podc} and the proof idea of \cref{thm:energy_lb_log_N_n}.

\begin{theorem}[Time lower bound of $\Omega(\min\{n, \frac{N}{2^k}\})$]\label{thm:timelb1}
Let $N$ be the size of the ID space. Let $\mathcal{A}$ be a deterministic leader election algorithm in the $\strongcd$ model that works for any set of devices $V$ of size at most $n$. Let $k$ denote the energy complexity of $\mathcal{A}$, and let $t$ denote the time complexity of $\mathcal{A}$. Then we have
\begin{align*}
    t \geq \Omega(\min\{n, \frac{N}{2^k}\}).
\end{align*}
\end{theorem}
\begin{proof}
\ifdefined\iscamera
For any $j \in [N]$ and $i \in [t]$, define $a_i(j) \in \{\idle$, $\listen$, $\transmit\}$ as the action of the device $v$ with $\ID(v) = j$ in the $i$th time slot if in the first $i-1$ time slots whenever $v$ listens it hears silence, and whenever $v$ transmits it detects collision. Using a similar argument as the proof of \cref{thm:energy_lb_log_N_n}, there exists a sequence $\{b_i\}_{i=1}^t$ that matches at least $\frac{N}{2^k}$ number of action sequences. Let $V$ denote the set of devices 
matching $\{b_i\}_{i=1}^t$. We have $|V| \geq \frac{N}{2^k}$. Let $V'$ be any subset of $V$ of size $\min \{n, \frac{N}{2^k}\}$. The algorithm $\mathcal{A}$ is correct on the set $V'$ and any subset of $V'$ since they all have size $\leq n$.

\else
For any $j \in [N]$ and $i \in [t]$, define $a_i(j) \in \{\idle, \listen, \transmit\}$ as the action of the device $v$ with $\ID(v) = j$ in the $i$th time slot if in the first $i-1$ time slots whenever $v$ listens it hears silence, and whenever $v$ transmits it detects collision.

Consider a random sequence $\{b_i\}_{i=1}^t$ where each $b_i$ is uniformly randomly sampled from $\{\listen, \transmit\}$. We say $\{b_i\}_{i=1}^t$ matches a sequence $\{a_i(j)\}_{i=1}^t$ if for any $i$, either $a_i(j) = \idle$, or $a_i(j) = b_i$. Since there are at most $k$ $\listen$ or $\transmit$ actions in the sequence $\{a_i(j)\}_{i=1}^t$, it is easy to see that
\[
\Pr[\{b_i\}_{i=1}^t \text{~matches~} \{a_i(j)\}_{i=1}^t] = \frac{1}{2^k}.
\]
Thus in expectation $\{b_i\}_{i=1}^t$ matches $\frac{N}{2^k}$ number of action sequences. This means there must exist some $\{b_i\}_{i=1}^t$ that matches at least $\frac{N}{2^k}$ number of action sequences. Let $V$ denote the set of devices that this $\{b_i\}_{i=1}^t$ matches with, and we have $|V| \geq \frac{N}{2^k}$. Let $V'$ be any subset of $V$ of size $\min \{n, \frac{N}{2^k}\}$. The algorithm $\mathcal{A}$ is correct on the set $V'$ and any subset of $V'$ since they all have size $\leq n$.
\fi

Next, we define a sequence of active sets $V' = A_1 \supseteq A_2 \supseteq \cdots \supseteq A_{|V'|} \neq \emptyset$ as follows: for each $A_{\ell}$, let $t_{\ell}$ be the minimum time slot such that there is exactly one device $v_{\ell} \in A_{\ell}$ with $a_{t_{\ell}}(j) = \transmit$ where $j = \ID(v_{\ell})$. Such a time slot $t_{\ell}$ must exist because the algorithm $\mathcal{A}$ is correct on set $A_{\ell}$. Recall that we consider an easier success criterion where a leader election algorithm is successful if there is a time slot where exactly one device transmits and at least one device listens.

Then we set $A_{\ell+1} = A_{\ell} \backslash \{v_{\ell}\}$. It remains to show that $t_{\ell} \neq t_{\ell'}$ for any $\ell' > \ell$. This means there exist $|V'|$ different time slots $t_1 \neq t_2 \neq \cdots \neq t_{|V'|}$, so we must have $t \geq |V'| = \min\{n, \frac{N}{2^k}\}$.

We prove $t_{\ell} \neq t_{\ell'}$ for $\ell' > \ell$ by contradiction. Suppose $t_{\ell} = t_{\ell'}$. When executing the algorithm $\mathcal{A}$ on set $A_{\ell}$, before time $t_{\ell}$, either all devices perform actions in $\{\listen, \idle\}$, or at least two devices $\transmit$, so a device $v$ in $A_{\ell}$ with ID $j$ will always perform the same action as the sequence $\{a_i(j)\}_{i=1}^{t_{\ell}}$. The same argument also holds for set $A_{\ell'}$. Since $A_{\ell'} \subseteq A_{\ell}$, we have $v_{\ell'} \in A_{\ell}$, so both $v_{\ell}$ and $v_{\ell'}$ will $\transmit$ on time $t_{\ell}$ when executing the algorithm $\mathcal{A}$ on set $A_{\ell}$. This contradicts the definition of $t_{\ell}$. Thus we must have $t_{\ell} \neq t_{\ell'}$. This finishes the proof.
\end{proof}

Using a  similar proof, we obtain the following corollary which considers algorithms that work for sets of devices of size \emph{at least} $n$, while 
in \cref{thm:timelb1} the algorithm works for sets of size at most $n$.
\begin{corollary}\label{cor:time_lb_N_2_k}
Let $N$ be the size of the ID space. Let $\mathcal{A}$ be a deterministic leader election algorithm in the $\strongcd$ model that works for any set of devices $V$ of size at least $n$. Let $k$ denote the energy complexity of $\mathcal{A}$, and let $t$ denote the time complexity of $\mathcal{A}$. Assume $k$ satisfies that $n < \frac{N}{2^{k+1}}$. Then we have
\begin{align*}
    t \geq \frac{N}{2^{k+1}}.
\end{align*}
\end{corollary}
\ifdefined\isspaa
We postpone the proof to \cref{sec:pf_time_lb_N_2_k}.
\else
\begin{proof}
We follow the proof of \cref{thm:timelb1} to get the set $V$ of size $|V| \geq \frac{N}{2^k}$ that a random sequence $\{b_i\}_{i=1}^t$ matches with. We define a sequence of active sets $V = A_1 \supseteq A_1 \supseteq \cdots \supseteq A_{|V|/2} \neq \emptyset$ similar as before: for each $A_{\ell}$, let $t_{\ell}$ be the minimum time slot such that there is exactly one device $v_{\ell} \in A_{\ell}$ with $a_{t_{\ell}}(j) = \transmit$ where $j = \ID(v_{\ell})$. Since $|A_{\ell}| = |A_{\ell-1}| - 1$, we have $|A_1| \geq \cdots \geq |A_{|V|/2}| \geq |V|/2 > n$, which means $\mathcal{A}$ is correct on all $A_{\ell}$, so $t_{\ell}$ must exist and this set sequence is well-defined.

Using a similar argument as in the proof of \cref{thm:timelb1}, we have $t_{\ell} \neq t_{\ell'}$ for any $\ell \neq \ell'$. This means there  exist $|V|/2$ different time slots $t_1 \neq t_2 \neq \cdots \neq t_{|V|/2}$, so we  have $t \geq |V|/2 \geq \frac{N}{2^{k+1}}$.
\end{proof}
\fi

\subsection{An  \texorpdfstring{$\Omega(\log \log \frac{N}{n})$}{Omega(log log N/n)} Energy Lower Bound for  \texorpdfstring{$\strongcd$ and $\sendercd$}{Strong-CD and Sender-CD}}
In this section we prove an energy lower bound of $\Omega(\log \log \frac{N}{n})$ in the $\strongcd$ model. This bound also holds in the weaker $\sendercd$ model. We first make a useful definition. 

\begin{definition}[Potential active time slot]\label{def:potential_active_time_slot}
Let $N$ be the size of the ID space. Let $\mathcal{A}$ be a deterministic leader election algorithm in the $\strongcd$ model. Let $t$ be the time complexity of $\mathcal{A}$. For any $j \in [N]$ and $i \in [t]$, we define $g_i(j) = 0$ if the device with $\ID=j$ is always idle in the $i$th time slot no matter what feedback it receives in the first $i-1$ time slots. Otherwise we define $g_i(j) = 1$. (Here we only consider feedback of $\noise$ or $\silence$, and no message is successfully transmitted.)

If $g_i(j) = 1$, we say the $i$th time slot is a potential active time slot for the device with $\ID=j$. 
\end{definition}

In \cref{cla:number_of_potential_active_time_slots} we bound the number of potential active time slots of a device when running an algorithm with energy complexity $k$.
\begin{claim}[Number of potential active time slots]\label{cla:number_of_potential_active_time_slots}
If the algorithm $\mathcal{A}$ has energy complexity $k$, then each of the devices has at most $2^{k}$ potential active time slots.
\end{claim}

\ifdefined\isspaa
We postpone the proof to \cref{sec:pf_number_of_potential_active_time_slots}.
\else
\begin{proof}
We use a protocol tree with height $t$ to describe all possible actions of a device $v$ when running the algorithm $\mathcal{A}$. If the action of a node is $\transmit$ or $\listen$, the feedback could be either $\noise$ or $\silence$, and the protocol tree has two branches following this node. 
If the action of a node is $\idle$, there is no feedback, and this node only has one child.

Since the device uses at most $k$ energy, there are at most $k$ $\transmit$ and $\listen$ actions along any path in the tree. Thus there are at most $\sum_{i=1}^{k} 2^{i-1} = 2^k - 1$ branching nodes. The number of potential active time slots is bounded by the number of branching nodes in the tree, so there are $\leq 2^k$ potential active time slots.
\end{proof}
\fi

We are ready to present the main lower bound 
of this section.
\begin{theorem}[Energy lower bound of $\Omega(\log \log \frac{N}{n})$]\label{thm:energy_lb_log_log_N_n}
Let $N$ be the size of the ID space. Let $\mathcal{A}$ be a deterministic leader election algorithm in the $\strongcd$ model that works for any set of devices $V$ of size at least $n$. Let $k$ denote the energy complexity of $\mathcal{A}$. Then we have
\begin{align*}
    k \geq \Omega(\log \log \frac{N}{n}).
\end{align*}
\end{theorem}
We prove this theorem using the following reduction.
\begin{lemma}[A reduction step for \cref{thm:energy_lb_log_log_N_n}]\label{lem:reduction_step}
Consider the leader election problem in the $\strongcd$ model. Suppose there exists a deterministic algorithm $\mathcal{A}$ for ID space of size $N$, where $\mathcal{A}$ is correct on any set $V \subseteq [N]$ of size at least $n$, and when running $\mathcal{A}$ each device uses at most $k$ energy and has at most $m$ potential active time slots where $m \leq 2^k$. Assume that $n$ satisfies $n \leq \frac{N}{2^{k+2}}$.

Then there exists another deterministic algorithm $\mathcal{A}'$ for ID space of size $\frac{N}{m \cdot 2^{k+6}}$, where $\mathcal{A}'$ is correct on any set $V' \subseteq [\frac{N}{m \cdot 2^{k+6}}]$ of size at least $n$, and when running $\mathcal{A}'$ each device uses at most $k$ energy and has at most $m-1$ potential active time slots. 
\end{lemma}
\ifdefined\isspaa
We postpone the proof to \cref{sec:pf_reduction_step}.
\else
\begin{proof}
Let $t$ denote the time complexity of $\mathcal{A}$, and let $g_i(j)$ for $i \in [t]$ and $j \in [N]$ be defined according to Definition~\ref{def:potential_active_time_slot}. 

We define $D \subseteq [t]$ as the set of time slots where more than $m \cdot 2^{k+3}$ devices are potentially active, i.e., $D$ includes all time slots $i \in [t]$ which satisfies $|\{j \in [N] ~|~ g_i(j) = 1\}| > m \cdot 2^{k+3}$. Note that we have $|D| \leq \frac{N \cdot m}{m \cdot 2^{k+3}} = \frac{N}{2^{k+3}}$ since each device has at most $m$ potential active time slots.

We define $V_1 \subseteq [N]$ as the set of devices that are only potentially active in time slots in $D$, i.e., $V_1$ includes all IDs $j \in [N]$ which satisfies $\{i \in [t] ~|~ g_i(j) = 1\} \subseteq D$.

We prove that $|V_1| \leq \frac{N}{2}$ by contradiction. Suppose $|V_1| > \frac{N}{2}$. We construct a deterministic leader election algorithm $\overline{\mathcal{A}}$ for an ID space of size $\frac{N}{2}$, and let each device run the same protocol as in the algorithm $\mathcal{A}$, while ignoring the time slots that are not in $D$. $\overline{\mathcal{A}}$ is correct on any set of devices of size at least $n$, and has energy complexity $k$ and time complexity $|D|$. Since we assumed $n \leq \frac{N/2}{2^{k+1}}$ in the lemma statement, Corollary~\ref{cor:time_lb_N_2_k} implies that $\overline{\mathcal{A}}$ satisfies $|D| \geq \frac{N/2}{2^{k+1}}$, and this contradicts with $|D| \leq \frac{N}{2^{k+3}}$ that we just proved.

Next we define a random set $V_2 \subseteq [N] \backslash V_1$ such that $V_2$ includes each $j \in [N] \backslash V_1$ independently with probability $\frac{1}{m \cdot 2^{k+4}}$. Then we define a set $V_3 \subseteq V_2$ that includes all $j \in V_2$ such that there exists a special time slot $i \in [t] \backslash D$ where $g_i(j) = 1$, and $g_i(j') = 0$ for all other $j' \in V_2 \backslash \{j\}$.

We prove that $\mathbb{E}[|V_3|] \geq \frac{N}{m \cdot 2^{k+6}}$. Consider a fixed $j \in [N] \backslash V_1$. By the definition of $V_1$ we know that there must exist at least one time slot $i \in [t] \backslash D$ such that $g_i(j) = 1$. Fix this time slot $i$, and define set $\overline{V} = \{j' \in N \backslash (V_1 \cup \{j\}) ~|~ g_i(j') = 1\}$. By the definition of $D$ and since $i \notin D$, we know that $|\overline{V}| \leq m \cdot 2^{k+3}$. So we have
\[
    \mathbb{E}_{V_2}[|\overline{V} \cap V_2|] = \frac{|\overline{V}|}{m \cdot 2^{k+4}} \leq \frac{1}{2}.
\]
Since $|\overline{V} \cap V_2|$ must be an integer, we have
$
    \Pr_{V_2}[|\overline{V} \cap V_2| = 0] \geq \frac{1}{2}.
$
Note that if $j \in V_3$, we must have $g_i(j') = 0$ for all $j' \in V_2 \backslash \{j\}$, which implies $|\overline{V} \cap V_2| = 0$, so we have
\begin{align*}
    \Pr_{V_2}[j \in V_3] = &~ \Pr_{V_2}[j \in V_3 ~|~ j \in V_2] \cdot \Pr_{V_2}[j \in V_2] \\
    \geq &~ \Pr_{V_2}[|\overline{V} \cap V_2| = 0] \cdot \Pr_{V_2}[j \in V_2]
    \geq \frac{1}{m \cdot 2^{k+5}}.
\end{align*}
Combining this bound on $\Pr_{V_2}[j \in V_3]$ and the previously proved bound $|V_1| \leq N/2$, we have
\[
\mathbb{E}_{V_2}[|V_3|] = \sum_{j \in [N] \backslash V_1} \Pr_{V_2}[j \in V_3] \geq \frac{N - |V_1|}{m \cdot 2^{k+5}} \geq \frac{N}{m \cdot 2^{k+6}}.
\]
Thus there must exist one instance of $V_3$ that satisfies $|V_3| \geq \frac{N}{m \cdot 2^{k+6}}$. 
By definition we know that for each $j \in V_3$ there exists a special time slot $i \in [t] \backslash D$ where $g_i(j) = 1$ and $g_i(j') = 0$ for all other $j' \in V_3$. Since in every special time slot only one device is potentially active, that device always gets $\silence$ so it can discard this action.

Finally we construct the new algorithm $\mathcal{A}'$ that works for an ID space of size $\frac{N}{m \cdot 2^{k+6}}$. Each device in $\mathcal{A}'$ runs the same protocols of $\mathcal{A}$ for devices in $V_3$, but it changes the action from $\{\transmit, \listen\}$ to $\idle$ in the special time slot. Thus we know that $\mathcal{A}'$ has at most $m - 1$ potential active time slots.
\end{proof}
\fi

\begin{proof}[Proof of \cref{thm:energy_lb_log_log_N_n}]
From \cref{cla:number_of_potential_active_time_slots} we know that in $\mathcal{A}$ each device has at most $2^k$ potential active time slots. We apply \cref{lem:reduction_step} 
$2^k$ times to construct algorithms $\mathcal{A}_1, \mathcal{A}_2, \ldots, \mathcal{A}_{2^k}$ where the algorithm $\mathcal{A}_{\ell}$ works for an ID space of size at least $\frac{N}{(2^{2k+6})^{\ell}}$ and each device has at most $2^k - \ell$ potential active time slots. Our goal is to prove that $n \geq \frac{N}{(2^{2k+6})^{2^k}}$. If for any intermediate $\ell < 2^k$ the condition $n \leq \frac{N}{(2^{2k+6})^{\ell} \cdot 2^{k+2}}$ is violated, then we already have the desired bound $k \geq \Omega(\log \log \frac{N}{n})$, so we can assume this condition is always satisfied, and the reduction step is always valid.

Now consider the algorithm $\mathcal{A}_{2^k}$ that works for ID space of size $\geq \frac{N}{(2^{2k+6})^{2^k}}$. Each device in $\mathcal{A}_{2^k}$ has zero potential active time slots, so this algorithm cannot be correct on any set of devices. However, we know that it is correct on set of devices of size at least $n$, thus we must have $n > \frac{N}{(2^{2k+6})^{2^k}}$, and this gives the desired bound $k \geq \Omega(\log \log \frac{N}{n})$.
\end{proof}

\section*{Acknowledgments}
We would like to thank Ruosong Wang for helpful discussions in the early stage of this work.


\bibliographystyle{alpha}
\bibliography{references}

\end{document}